%% file: fast-cograph-completion.tex
\documentclass[a4paper]{llncs}

\usepackage{epic,eepic}
\usepackage{graphicx,color}

\usepackage{amssymb}
\usepackage{amsmath}
\usepackage{amsfonts} 
\usepackage{stmaryrd} 
\usepackage{latexsym} 

\newtheorem{remarknum}{Remark}

\renewenvironment{proof}{\noindent\textbf{Proof.}}{\mbox{}\hfill $\Box$\\}

\newcommand\set[1]{\ensuremath{\{ #1 \} }}
\renewcommand\int[1]{\ensuremath{\llbracket#1\rrbracket}}
\newcommand\cur[1]{\ensuremath{{\cal{#1}}}}
\newcommand\commentaire[1]{}

\newcommand{\cutrank}{\mbox{\rm cutrank}}

\title{Faster and Enhanced Inclusion-Minimal Cograph Completion
\thanks{This project has received funding from the European Union's Horizon 2020 research and innovation programme under the Marie Sklodowska-Curie grant agreement No 749022 and has been supported by Pareto-Optimal Parameterized Algorithms, ERC Starting Grant 715744.}
\thanks{This work was partially funded by the PICS program of CNRS, by the Vietnam National Foundation for Science and Technology Development (NAFOSTED) under the grant number 101.99-2016.16 and by the Vietnam Institute for Advanced Study in Mathematics (VIASM).}
}

\author{Christophe Crespelle\inst{1} \and Daniel Lokshtanov\inst{1} \and Thi Ha Duong Phan\inst{2} \and Eric Thierry\inst{3}}

\authorrunning{C. Crespelle \and D. Lokshtanov \and T.H.D. Phan \and E. Thierry}

\institute{
University of Bergen, Department of Informatics, N-5020 Bergen, NORWAY\\
\email{christophe.crespelle@uib.no, daniello@ii.uib.no}
\and
Institute of Mathematics, Vietnam Academy of Science and Technology,\\ 18 Hoang Quoc Viet, Hanoi, Vietnam\\
\email{phanhaduong@math.ac.vn}
\and
Univ Lyon, ENS de Lyon, UCB Lyon 1, CNRS, Inria, LIP UMR 5668,\\ 15 parvis Ren{\'e} Descartes, F-69342, Lyon, FRANCE\\
\email{eric.thierry@ens-lyon.fr}
}

\begin{document}

\maketitle




\begin{abstract}
\begin{sloppypar}
We design two incremental algorithms for computing an inclusion-minimal completion of an arbitrary graph into a cograph. The first one is able to do so while providing an additional property which is crucial in practice to obtain inclusion-minimal completions using as few edges as possible : it is able to compute a minimum-cardinality completion of the neighbourhood of the new vertex introduced at each incremental step. It runs in $O(n+m')$ time, where $m'$ is the number of edges in the completed graph. This matches the complexity of the algorithm in~\cite{LMP10} and positively answers one of their open questions.
Our second algorithm improves the complexity of inclusion-minimal completion to $O(n+m\log^2 n)$ when the additional property above is not required. Moreover, we prove that many very sparse graphs, having only $O(n)$ edges, require $\Omega(n^2)$ edges in any of their cograph completions. For these graphs, which include many of those encountered in applications, the improvement we obtain on the complexity scales as $O(n/\log^2 n)$.
\end{sloppypar}
\end{abstract}

\section{Introduction}

We consider the problem of completion of an arbitrary graph into a \emph{cograph}, i.e. a graph with no induced path on $4$ vertices. This is a particular case of \emph{graph modification problem}, in which one wants to perform elementary modifications to an input graph, typically adding and removing edges and vertices, in order to obtain a graph belonging to a given target class of graphs, which satisfies some additional property compared to the input. Ideally, one would like to do so by performing a minimum number of elementary modifications. This is a fundamental problem in graph algorithms, which 
corresponds to the notion of projection in geometry: given an element $a$ of a ground set $X$ equipped with a distance and a subset $S\subseteq X$, find an element of $S$ that is closest to $a$ for the provided distance (here, the number of elementary modifications performed on the graph). This is also the meaning of modification problems in algorithmic graph theory: they answer the question to know how far is a given graph from satisfying a target property.


Here, we consider the modification problem called \emph{completion}, where only one operation is allowed: adding an edge. In this case, the quantity to be minimised, called the \emph{cost} of the completion, is the number of edges added, which are called \emph{fill edges}. The particular case of completion problems has been shown very useful in algorithmic graph theory and several other contexts. These problems are closely related to some important graph parameters, such as treewidth~\cite{ACP87}, and can help to efficiently solve problems that otherwise are hard on the input graph~\cite{BK06}. They are also useful for other algorithmic problems arising in computer science, such as sparse matrix multiplication~\cite{R72}, and in other disciplines such as archaeology~\cite{K69}, molecular biology~\cite{BDFHW95} and genomics, where they played a key role in the mapping of the human genome~\cite{GGK+95,K93}.


Unfortunately, finding the minimum number of edges to be added in a completion problem is NP-hard for most of the target classes of interest (see, e.g., the thesis of Mancini~\cite{Man08} for further discussion and references). To deal with this difficulty of computation, the domain has developed a number of approaches. This includes approximation~\cite{NSS00}, restricted input~\cite{BKK+98,BT01,BDK00,KKW98,KKS97,Mei10}, parameterization~\cite{Cai96,DGH+05,KST04,Man10,VHP+09} and inclusion-minimal completions. In the latter approach, one does not ask for a completion having the minimum number of fill edges but only ask for a set of fill edges which is minimal for inclusion, i.e. which does not contain any proper subset of fill edges whose addition also results in a graph in the target class. This is the approach we follow here. In addition to the case of cographs~\cite{LMP10}, it has been followed for many other graph classes, including chordal graphs~\cite{HTV05}, interval graphs~\cite{CT13,OMKF81}, proper interval graphs~\cite{RST08}, split graphs~\cite{HM09}, comparability graphs~\cite{HMP08} and permutation graphs~\cite{CPT15}.

\begin{sloppypar}
The rationale behind the inclusion-minimal approach is that minimum-cardinality completions are in particular inclusion-minimal. Therefore, if one is able to sample\footnote{Usually, minimal completion algorithms are not fully deterministic. There are some choices to be made arbitrarily along the algorithm and different choices lead to different minimal completions.} efficiently the space of inclusion-minimal completions, one can compute several of them, pick the one of minimum cost and hope to get a value close to the optimal one. One of the reason of the success of inclusion-minimal completion algorithms is that this heuristic approach was shown to perform quite well in practice~\cite{BHS03,BHT01}. The second reason of this success, which is a key point for the approach, is that it is usually possible to design algorithms of low complexity for the inclusion-minimal relaxation of completion problems.
\end{sloppypar}

\subsubsection{Related work.}
Modification problems into the class of cographs have already received a great amount of attention~\cite{GHP12,HFW+15,HWL+15,LWG+12,LMP10}, as well as modification problems into some of its subclasses, such as \emph{quasi-threshold graphs}~\cite{BHS+15} and \emph{threshold graphs}~\cite{DDL+15}. One reason for this is that cographs are among the most widely studied graph classes. They have been discovered independently in many contexts~\cite{CLS81} and they are known to admit very efficient algorithms for problems that are hard in general~\cite{BLS99}
. Moreover, very recently, cograph modification was shown a powerful approach to solve problems arising in complex networks analysis, e.g. community detection~\cite{JGG+15}, inference of phylogenomics~\cite{HWL+15} and modelling~\cite{Cre17}. The modification problem into the class of quasi-threshold graphs has also been used and it revealed that complex networks encountered in some contexts are actually very close to be quasi-threshold graphs~\cite{BHS+15}, in the sense that only a few modifications are needed to transform them into quasi-threshold graphs.
This growing need for treating real-world datasets, whose size is often huge, asks for more efficient algorithms both with regard to the running time and with regard to the quality (number of modifications) of the solution returned by the algorithm.

\subsubsection{Our results.}

Our main contribution is to design two algorithms for inclusion-minimal cograph completion. The first one (Section~\ref{sec:algo-compl}) is an improvement of the incremental algorithm in~\cite{LMP10}. It runs in the same $O(n+m')$ complexity, where $m'$ is the number of edges in the completed graph, and is in addition able to select one minimum-cardinality completion of the neighbourhood of the new incoming vertex at each incremental step of the algorithm, which is an open question in~\cite{LMP10} (Question~3 in the conclusion) which we positively answer here. It must be clear that this does not guarantee that the completion computed at the end of the algorithm has minimum cardinality 
but this feature is highly desirable in practice to 
obtain completions using as few fill edges as possible.

When this additional feature is not required, our second algorithm (Section~\ref{sec:improv-algo}) solves the inclusion-minimal problem in $O(n+m\, log^2 n)$ time, which only depends on the size of the input. Furthermore, we prove that many sparse graphs, namely those having mean degree fixed to a constant, require $\Omega(n^2)$ edges in any of their cograph completions. This result is worth of interest in itself and implies that, for such graphs, which have only $O(n)$ edges, the improvement of the complexity we obtain with our second algorithm is quite significant : a factor $n/log^2 n$.

\section{Preliminaries}\label{sec:prel}

All graphs considered here are finite, undirected, simple and loopless. In the following, $G$ is a graph, $V$ (or $V(G)$) is its vertex set and $E$ (or $E(G)$) is its edge set. We use the notation $G=(V,E)$, $n=|V|$ stands for the cardinality of $V$ and $m=|E|$ for the cardinality of $E$.
An edge between vertices $x$ and $y$ will be arbitrarily denoted by $xy$ or $yx$. The neighbourhood of $x$ is denoted by $N(x)$ (or $N_G(x)$) and for a subset $X\subseteq V$, we define $N(X)=(\bigcup_{x\in X} N(x))\setminus X$. The subgraph of $G$ induced by some $X \subseteq V$ is denoted by $G[X]$. 

For a rooted tree $T$ and a node $u\in T$, we denote $parent(u)$, $\cur{C}(u)$, $Anc(u)$ and $Desc(u)$ the {\em parent} and the set of {\em children}, {\em ancestors} and {\em descendants} of $u$ respectively, using the usual terminology and with $u$ belonging to $Anc(u)$ and $Desc(u)$. The \emph{lowest common ancestor} of two nodes $u$ and $v$, denoted $lca(u,v)$, is the lowest node in $T$ which is an ancestor of both $u$ and $v$.
The subtree of $T$ rooted at $u$, denoted by $T_u$, is the tree induced by node $u$ and all its descendants in $T$.
We use two other notions of subtree, which we call \emph{upper tree} and \emph{extracted tree}. The upper tree of a subset of nodes $S$ of $T$ is the tree, denoted $T^{up}_S$, induced by the set $Anc(S)$ of all the ancestors of the nodes of $S$, i.e. $Anc(S)=\bigcup_{s\in S} Anc(s)$.
The tree extracted from $S$ in $T$, denoted $T^{xtr}_S$, is defined as the tree whose set of nodes is $S$ and whose parent relationship is the transitive reduction of the ancestor relationship in $T$. More explicitly, for $u,v\in S$, $u$ is the parent of $v$ in $T^{xtr}_S$ iff $u$ is an ancestor of $v$ in $T$ and there exist no node $v'\in S$ such that $v'$ is a strict ancestor of $v$ and a strict descendant of $u$ in $T$.

\begin{figure}
  \begin{center}
  $
    \vcenter{\hbox{\includegraphics[scale=0.9]{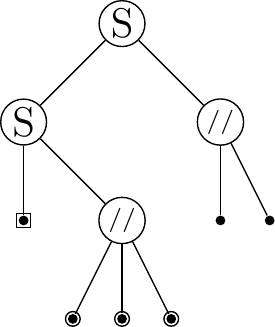}}}
    \hspace{1cm}
    \vcenter{\hbox{\includegraphics[scale=0.9]{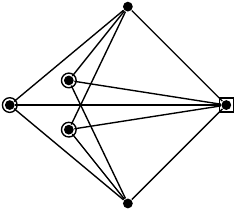}}}
    \hspace{1cm}
    \vcenter{\hbox{\includegraphics[scale=0.9]{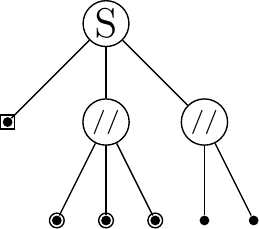}}}
  $
  \caption{Example of a labelled construction tree (left), the cograph it represents (centre), and the associated cotree (right). Some vertices are decorated in order to ease the reading.}
  \label{fig:cograph}
  \end{center}
\end{figure}

\subsubsection{Cographs.}
One of their simpler definitions is that they are the graphs that do not admit the $P_4$ (path on $4$ vertices) as an induced subgraph. This shows that the class is \emph{hereditary}, i.e., an induced subgraph of a cograph is also a cograph. Equivalently, they are the graphs obtained from a single vertex under the closure of the \emph{parallel} composition and the \emph{series} composition.
The parallel composition of two graphs $G_1=(V_1,E_1)$ and $G_2=(V_2,E_2)$ is their disjoint union, i.e., the graph $G_{par}=\big(V_1 \cup V_2, E_1 \cup E_2\big)$.
The series composition of $G_1$ and $G_2$ is their disjoint union plus all possible edges between vertices of $G_1$ and vertices of $G_2$, i.e., the graph $G_{ser}=\big(V_1 \cup V_2, E_1 \cup E_2 \cup \{xy \mid x \in V_1, y \in V_2\}\big)$. These operations can naturally be extended to an arbitrary finite number of graphs.

This gives a nice representation of a cograph $G$ by a tree whose leaves are the vertices of $G$ and whose internal nodes (non-leaf nodes) are labelled $//$, for parallel, or $S$, for series, corresponding to the operations used in the construction of $G$. It is always possible to find such a labelled tree $T$ representing $G$ such that every internal node has at least two children, no two parallel nodes are adjacent in $T$ and no two series nodes are adjacent. This tree $T$ is unique~\cite{CLS81} and is called the \emph{cotree} of $G$, see example in Fig.~\ref{fig:cograph}.
Note that the subtree $T_u$ rooted at some node $u$ of cotree $T$ also defines a cograph, denoted $G_u$, whose set of vertices is the set of leaves of $T_u$, denoted $V(u)$ in the following.
The adjacencies between vertices of a cograph can easily be read on its cotree, in the following way.

\begin{remarknum}\label{rem:adjcotree}
Two vertices $x$ and $y$ of a cograph $G$ having cotree $T$ are adjacent iff the lowest common ancestor $u$ of leaves $x$ and $y$ in $T$ is a series node. Otherwise, if $u$ is a parallel node, $x$ and $y$ are not adjacent.
\end{remarknum}


\subsubsection{The incremental approach.}
Our approach for computing a minimal cograph completion of an arbitrary graph $G$ is incremental, in the sense that we consider the vertices of $G$ one by one, in an arbitrary order $(x_1, \dots, x_n)$, and at step $i$ we compute a minimal cograph completion $H_i$ of $G_i = G[\{x_1, \dots, x_i\}]$ from a minimal cograph completion $H_{i-1}$ of $G_{i-1}$, by adding only edges incident to $x_i$. This is possible thanks to the following observation that is general to all hereditary graph classes that are also stable by addition of a universal vertex, which holds in particular for cographs.

\begin{lemma}[see e.g. \cite{OMKF81}]\label{lem:heruni}
Let $G=(V,E)$ be an arbitrary graph and let $H$ be a minimal cograph completion of $G$. Consider a new vertex $x\not\in V$ adjacent to an arbitrary subset $N(x)\subseteq V$ of vertices and denote $G'=G+x$ and $H'=H+x$ the graphs obtained by adding $x$ to $G$ and $H$ respectively. Then, there exists a subset $M\subseteq V\setminus N(x)$ of vertices such that $H''=(V,E(H')\cup\set{xy\ |\ y\in M})$ is a cograph. Moreover, for any such set $M$ which is minimal for inclusion, $H''$ is an inclusion-minimal cograph completion of $G'$. We call such completions \emph{(minimal) constrained completions} of $G+x$.
\end{lemma}

For any subset $S\subseteq V$ of vertices, we say that we \emph{fill} $S$ in $H''$ if we make all the vertices of $S\setminus N(x)$ adjacent to $x$ in the completion $H''$ of $G+x$. The edges added in a completion are called \emph{fill edges} and the \emph{cost} of the completion is its number of fill edges.

\subsubsection{The new problem.}
From now on, we consider the following problem, with slightly modified notations.
$G=(V,E)$ is a cograph, and $G+x$ is the graph obtained by adding to $G$ a new vertex $x$ adjacent to some arbitrary subset $N(x)$ of vertices of $G$. Both our algorithms take as input the cotree of $G$ and the neighbourhood $N(x)$ of the new vertex $x$. They compute the set $N'(x)\supseteq N(x)$ of neighbours of $x$ in some minimal constrained cograph completion $H$ of $G+x$, i.e. obtained by adding only edges incident to $x$ (cf. Lemma~\ref{lem:heruni}). Then, the cotree of $G$ is updated under the insertion of $x$ with neighbourhood $N'(x)$, in order to obtain the cotree of $H$ which will serve as input in the next incremental step.

We now introduce some definitions and characterisations we use in the following.

\begin{definition}[Full, hollow, mixed]
Let $G$ be a cograph and let $x$ be a vertex to be inserted in $G$ with neighbourhood $N(x)\subseteq V(G)$. A subset $S\subseteq V(G)$ is \emph{full} if $S\subseteq N(x)$, \emph{hollow} if $S\cap N(x)=\varnothing$ and \emph{mixed} if $S$ is neither full nor hollow. When $S$ is full or hollow, we say that $S$ is \emph{uniform}.
\end{definition}

We use these notions for nodes $u$ of the cotree as well, referring to their associated set of vertices $V(u)$. We denote $\cur{C}_{nh}(u)$ the subset of non-hollow children of a node $u$.

Theorem~\ref{th:insert} below gives a characterisation of the neighbourhood of a new vertex $x$ so that $G+x$ is a cograph. 

\begin{theorem}[\cite{CPS85,CP06}]\label{th:insert}(Cf. Fig.~\ref{fig:carac})
Let $G$ be a cograph with cotree $T$ and let $x$ be a vertex to be inserted in $G$ with neighbourhood $N(x)\subseteq V(G)$. If the root of $T$ is mixed, then $G+x$ is a cograph iff there exists a mixed node $u$ of $T$ such that:
\begin{enumerate}
\item all children of $u$ are uniform and
\item for all vertices $y\in V(G)\setminus V(u)$, $y\in N(x)$ iff $lca(y,u)$ is a series node.
\end{enumerate}
Moreover, when such a node $u$ exists, it is unique and it is called the \emph{insertion node}.
\end{theorem}

\begin{remarknum}\label{rem:rootnofull}
In all the rest of the article, we do not consider the case where the new vertex $x$ is adjacent to none of the vertices of $G$ or to all of them. Therefore, the root of the cotree $T$ of $G$ is always mixed wrt. $x$.
\end{remarknum}

\begin{figure}
\begin{center}
  \includegraphics[width=0.5\textwidth]{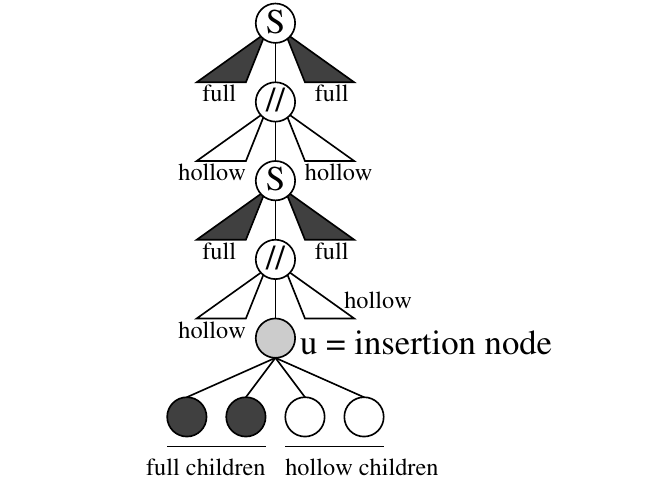}
  \caption{Illustration of Theorem~\ref{th:insert}: characterisation of the neighbourhood of a new vertex $x$ so that $G+x$ is a cograph. The nodes and triangles in black (resp. white) correspond to the parts of the tree that are full wrt. $x$ (resp. hollow wrt. $x$). The insertion node $u$, which is mixed, appears in grey colour.}
  \label{fig:carac}
\end{center}
\end{figure}

The reason for this is that the case where the root is uniform is straightforward: the only minimal completion of $G+x$ adds an empty set of edges and the update of cotree $T$ is very simple. By definition, inserting $x$ in $G$ with its neighbourhood $N'(x)$ in some constrained cograph completion $H$ of $G+x$ results in a cograph, namely $H$. Therefore, to any such completion $H$ we can associate one insertion node which is uniquely defined, from Theorem~\ref{th:insert} and from the restriction stated in Remark~\ref{rem:rootnofull}.

\begin{definition}
Let $G$ be a cograph with cotree $T$ and let $x$ be a vertex to be inserted in $G$. A node $u$ of $T$ is called a \emph{completion-minimal insertion node} iff there exists a minimal constrained completion $H$ of $G+x$ such that $u$ is the insertion node associated to $H$.
\end{definition}

From now and until the end of the article, $G$ is a cograph, $T$ is its cotree, $x$ is a vertex to be inserted in $G$ and we consider only constrained cograph completions of $G+x$. We therefore omit to systematically precise it.

\section{Characterisation of minimal constrained completions}\label{sec:charac}

The goal of this section is to give necessary and sufficient conditions for a node $u$ of $T$ to be a completion-minimal insertion node. From Theorem~\ref{th:insert}, the subtrees attached to the parallel strict ancestors of the insertion node $u$ must be hollow. As we can modify the neighbourhood of $x$ only by adding edges, it follows that if $u$ is the insertion node of some completion, then $u$ is \emph{eligible}, as defined below.

\begin{definition}[eligible]\label{def:eligible}
A node $u$ of $T$ is \emph{eligible} iff for all the strict ancestors $v$ of $u$ that are parallel nodes, all the children of $v$ distinct from its unique child $u'\in \cur{C}(v)\cap Anc(u)$ are hollow.
\end{definition}

When a node $u$ is eligible, there is a natural way to obtain a completion of the neighbourhood of $x$, which we call the completion anchored at $u$.

\begin{definition}[Completion anchored at $u$]\label{def:anchor}
Let $u$ be an eligible node of $T$. The \emph{completion anchored at $u$} is the one obtained by making $x$ adjacent to all the vertices of $V(G)\setminus V(u)$ whose lowest common ancestor with $u$ is a series node and by filling all the children of $u$ that are non-hollow.
\end{definition}

The completion anchored at some eligible node $u$ may not be minimal but, on the other hand, all minimal completions $H$ are completions anchored at some eligible node $u$, namely the insertion node of $H$.

\begin{lemma}\label{lem:uniq-compl}
For any completion-minimal insertion node $u$ of $T$, there exists a unique minimal completion $H$ of $G+x$ such that $u$ is the insertion node associated to $H$ and this unique completion is the completion anchored at $u$.
\end{lemma}

\begin{proof} 
First, note that the modified neighbourhood $N'_{\bar{u}}(x)$ of $x$ in $V(G)\setminus V(u)$ is given by Theorem~\ref{th:insert} and is the same for every completion having $u$ as insertion node. Moreover, as in any such completion, the children of $u$ in $T$ are uniform, then any non-hollow child $v$ of $u$ must be filled. Then, the completion $H_{min}$ defined by the modified neighbourhood $N'_{min}(x)=N'_{\bar{u}}(x)\cup \bigcup_{v\in\cur{C}(u) \text{ and } v \text{ is non-hollow}} V(v)$ of $x$ is included in every completion having $u$ as insertion node. As there exists some minimal completion having $u$ as insertion node, then from Theorem~\ref{th:insert}, $u$ is left mixed after completion and so $u$ has some hollow child with regard to $N(x)$. Consequently, $u$ is also mixed with regard to $N'_{min}(x)$. Finally, since the insertion of $x$ with neighbourhood $N'_{min}(x)$ satisfies conditions~1 and~2 of Theorem~\ref{th:insert}, then the completion $H'_{min}$ has $u$ as insertion node. And since $H'_{min}$ is included in all such completions, it follows that $H'_{min}$ is the unique minimal completion having $u$ as insertion node.
\end{proof}


To characterise completion-minimal insertion nodes, we will use the notion of \emph{forced nodes}. Their main property (see Lemma~\ref{lem:forced} below) is that they are full in any completion of $G+x$.

\begin{definition}[Completion-forced]\label{def:compl-forced}
Let $G$ be a cograph with cotree $T$ and let $x$ be a vertex to be inserted in $G$. A \emph{completion-forced} (or simply \emph{forced}) node $u$ is inductively defined as a node satisfying at least one of the three following conditions:
\begin{enumerate}
\item $u$ is full, or
\item $u$ is a parallel node with all its children non-hollow, or
\item $u$ is a series node with all its children completion-forced.
\end{enumerate}
\end{definition}

\begin{lemma}\label{lem:forced-fort}
Let $G$ be a cograph with cotree $T$ and let $x$ be a vertex to be inserted in $G$. A node $u$ of $T$ is completion-forced iff there exists a unique cograph completion of $G_u+x$, which is the one where all missing edges between $x$ and $V(u)$ are added.
\end{lemma}

\begin{proof}
Let us show the result by induction on $|V(u)|$.
First, consider a completion-forced node $u$ of $T$ and a completion $H$ of $G_u+x$. If $u$ satisfies Condition~3 of Definition~\ref{def:compl-forced}, then, by induction hypothesis, all its children are full in $H$ (as $H$ is also a cograph completion of $G_v+x$, for any child $v$ of $u$) and so is $u$. If $u$ satisfies Condition~1, then since $u$ is full before completion, it is also full after. Consider now the case where $u$ is completion-forced because it satisfies Condition~2 of Definition~\ref{def:compl-forced}, i.e. $u$ is parallel and all its children are non-hollow.

Assume for contradiction that $H$ does not fill $u$. Then, denote $u'$ the insertion node associated to $H$ in $T_u$. Theorem~\ref{th:insert} implies that $u'$ is eligible, and since all the children of $u$ are non hollow, it follows that $u'$ is not a strict descendant of $u$. Consequently, $u'=u$ and since all the children of $u$ are non hollow, Lemma~\ref{lem:uniq-compl} implies that $H$ fills all of them, and so $H$ fills $u$ as well: contradiction. Thus, $u$ is filled in any completion $H$ of $G_u+x$ and therefore, there exists a unique such completion.

Conversely, consider a non-completion-forced node $u$ of $T$. If $u$ is a series node, then $u$ has at least one non-completion-forced child $v$. By induction hypothesis, there exists a completion $H'$ of $G_v+x$ that does not fill $v$. Then, the completion $H$ of $G_u+x$ that coincides with $H'$ on $V(v)$ and that fills all the other children of $u$ is a cograph completion of $G_u+x$ that does not fill $u$. Now, if $u$ is a parallel node, then $u$ has at least one hollow child $v$. As $u$ is clearly eligible in $T_u$, the cograph completion $H$ anchored at $u$ is properly defined. Since $H$ leaves $v$ hollow, then $H$ does not fill $u$, which achieves the proof.
\end{proof}

\begin{lemma}\label{lem:forced}
Any completion-forced node $u$ of $T$ is filled in all the completions of $G+x$.
\end{lemma}

\begin{proof} 
This is a direct consequence of Lemma~\ref{lem:forced-fort}. Indeed, any completion of $G+x$ restricted to $V(u)$ is a completion of $G_u+x$. Moreover, from Lemma~\ref{lem:forced-fort}, there exists a unique cograph completion of $G_u+x$ and this completion makes $V(u)$ full.
\end{proof}

The next remark directly follows from Theorem~\ref{th:insert} and Lemma~\ref{lem:uniq-compl}.

\begin{remarknum}\label{rem:insert-mixed}
The insertion node $u$ of any minimal completion of $G+x$ has at least one hollow child and at least one non-hollow child. Therefore, $u$ is non-hollow and non-completion-forced.
\end{remarknum}


We now characterise the nodes $u$ that contain some minimal-insertion node in their subtree $T_u$ (including $u$ itself). In our algorithms, we will use this characterisation to decide whether we have to explore the subtree of a given node.

\begin{lemma}\label{lem:explore}
For any node $u$ of $T$, $T_u$ contains some completion-minimal insertion node iff $u$ is eligible, non-hollow and non-completion-forced.
\end{lemma}

\begin{proof} 
If $u$ is eligible non-hollow and non-completion-forced, consider such a node $v$ of $T_u$ which is lower possible in $T_u$. If $v$ is a series node, as $v$ is eligible so are all its children. It follows that all the children of $v$ are either completion-forced or hollow. Since $v$ is non-completion-forced, at least one of its children is hollow and since $v$ is non-hollow at least one of its children is non-hollow.
The same holds if $v$ is a parallel node: since $v$ is non-completion-forced, at least one of its children is hollow and since $v$ is non-hollow at least one of its children is non-hollow. Then, in both cases, in the completion $H'$ anchored at $v$, $v$ is mixed and so is $u$. Consequently, there exists a minimal completion $H$ included in $H'$ and necessarily $u$ is mixed in $H$ as well. From Theorem~\ref{th:insert}, it is straightforward to see that all minimal completions having an insertion node out of $T_u$ leaves $u$ full or hollow. It follows that the insertion node associated to $H$ belongs to $T_u$.

Now, conversely, if there exists $v\in T_u$ which is a completion-minimal insertion node, let us denote $H$ the minimal completion anchored at $v$. From Remark~\ref{rem:insert-mixed}, $v$ is non-hollow in $G+x$, and so is $u$. Moreover, from Theorem~\ref{th:insert}, it is straightforward to see that $v$ is eligible and so is $u$. From Theorem~\ref{th:insert} again, $v$ is mixed in $H$ and so is $u$. Then, Lemma~\ref{lem:forced} implies that $u$ is non-completion-forced, which achieves the proof of the lemma.
\end{proof}

Lemma~\ref{lem:stop} below gives additional conditions for $u$ itself to be an insertion node.


\begin{lemma}\label{lem:stop}
A node $u$ of $T$ is a completion-minimal insertion node iff $u$ is eligible, non-hollow and non-completion-forced and $u$ satisfies in addition one of the two following conditions:
\begin{enumerate}
\item\label{ser} $u$ is a series node and $u$ has at least one hollow child, or
\item\label{par} $u$ is a parallel node and $u$ has no eligible non-completion-forced child.
\end{enumerate}
\end{lemma}

\begin{proof} 
We first show that if the conditions of the lemma are satisfied, then $u$ is a completion-minimal insertion node.
From Lemma~\ref{lem:uniq-compl}, if $u$ is a completion-minimal insertion node, then there exists a unique minimal completion $H$ such that $u$ is the insertion node associated to this completion. From Lemma~\ref{lem:uniq-compl} again, this completion $H$ is the completion anchored at $u$, which is properly defined here as $u$ is eligible, see Definition~\ref{def:anchor}. We will now show that $H$ is minimal.

If $u$ is a parallel node, as $u$ is non-completion-forced, $u$ has at least one hollow child $v$, and the same holds if $u$ is a series node because of Condition~1. From Definition~\ref{def:anchor}, $v$ is hollow in $H$. Let $H'$ be a minimal completion of $G+x$ and let $u'$ be its insertion node $u'$. We will show that $H'$ is not strictly included in $H$. From Lemma~\ref{lem:uniq-compl}, if $u'=u$, then $H'=H$ and therefore, from now, we consider only the case where $u'\neq u$. Note that, from Theorem~\ref{th:insert}, the only nodes of $T$ that remain mixed after completion into $H'$ are the ancestors of $u'$. All the non-hollow nodes of $T$ that are not ancestors of $u'$ are filled in $H'$. Then, if $u'$ is not a descendant of $u$, node $u$ is filled in $H'$ and so is node $v$. It follows that, if $u'$ is not a descendant of $u$, $H'$ is not included in $H$.

Now, consider the case where $u'$ is a strict descendant of $u$ (remember that $u'\neq u$) and suppose for contradiction that $u$ is a parallel node. Lemma~\ref{lem:explore} implies that $u'$ is eligible. Since $u'$ is a strict descendant of $u$, then all the children of $u$, except its child $w$ that is an ancestor of $u'$, are hollow. Then, from Condition~2 of the present lemma, it follows that $w$ must be completion-forced. Lemma~\ref{lem:forced} implies that $w$, and so $u'$, is filled in $H'$. This contradicts the fact that $u'$ is the insertion node, as from Theorem~\ref{th:insert}, this node remains mixed after completion. Thus, $u$ is not a parallel node, but a series node.
From Remark~\ref{rem:insert-mixed}, $u'$ is non-hollow in $G+x$ and consequently, $u'$ is not a descendant of $v$ (the hollow child of $u$). Since $u$ is a series node, it follows that $v$ is filled in $H'$, which is therefore not included in $H$. This achieves the proof that the conditions of the lemma are sufficient.

Let us now show that they are necessary. Consider a completion-minimal node $u$ and let us show that it satisfies the conditions of the lemma. Firstly, because $T_u$ contains some completion-minimal insertion node, namely $u$, Lemma~\ref{lem:explore} implies that $u$ is mixed, eligible and non-completion-forced. Let $H$ be the completion anchored at $u$. From Theorem~\ref{th:insert}, $u$ is mixed in $H$. Then, from Lemma~\ref{lem:uniq-compl}, it follows that $u$ has at least one hollow child. Condition~1 is satisfied.

We now show that if $u$ is parallel and does not satisfy Condition~2, then the completion $H$ anchored at $u$ is not minimal, which implies that $u$ is not a completion-minimal insertion node. Since $u$ is mixed, it has at least one non-hollow child $v$. Moreover, since $u$ does not satisfy Condition~2, $v$ is the unique non-hollow child of $u$ (then $v$ is eligible) and $v$ is non-completion-forced. As $v$ is eligible, non-hollow and non-completion-forced, it follows from Lemma~\ref{lem:explore} that $T_v$ contains some completion-minimal insertion node. The corresponding minimal completion $H'$ is included in $H$ and even strictly included as $H'$ leaves $v$ mixed, while $H$ fills it (since $v$ is not hollow). Thus, $H$ is not minimal. By contraposition, if $H$ is minimal, Condition~2 is satisfied. This achieves the proof of the lemma.
\end{proof}

\section{An $\mathbf{O(n+m')}$ algorithm with incremental minimum}\label{sec:algo-compl}

In this section
, we design an incremental algorithm whose overall time complexity is $O(n+m')$, where $m'$ is the number of edges in the output completed cograph. We concentrate on one incremental step, whose input is the cotree $T$ of some cograph $G$ (the completion computed so far) and a new vertex $x$ together with the list of its neighbours $N(x)\subseteq V(G)$. Each node $u\in T$ stores its number $|\cur{C}(u)|$ of children and the number $|V(u)|$ of leaves in $T_u$. One incremental step takes time $O(d')$, where $d'$ is the degree of $x$ in the completion of $G+x$ computed by the algorithm. Within this complexity, our algorithm scans all the minimal completions of the neighbourhood of $x$ and select one of minimum cardinality. Our description is in two steps.

\subsubsection{First step: collecting information on nodes of $\mathbf{T}$.}

In this step, for each non-hollow node $u$ of $T$ we determine the following information: i) the list of its non-hollow children $\cur{C}_{nh}(u)$, ii) the number of neighbours of $x$ in $V(u)$ and iii) whether it is completion forced or not. To this purpose, we perform two bottom-up searches of $T$ from the leaves of $T$ that are in $N(x)$ up until the root of $T$. Consequently, each of these searches discovers exactly the set $\cur{NH}(T)$ of non-hollow nodes of $T$ (for which we show later that their number is $O(d')$).

In the first search, we label each node encountered as non-hollow, we build the list of its non-hollow children and count them. The nodes that are not visited, and therefore not labelled are exactly the hollow nodes of $T$.

In the second search, for each non-hollow node $u$ we determine the rest of its information, that is ii) the number of neighbours of $x$ in $V(u)$ and iii) whether it is completion forced or not. 

It is straightforward to get this information for the leaves $l$ of $T$ that belong to $N(x)$: there is exactly one neighbour of $x$ in $V(l)$ and $l$ is forced. Then, all the leaves in $N(x)$ forward their information to their parents in an asynchronous way. Along this process, each non-hollow node $u$ of $T$ is able to know whether it has received the information from all its non-hollow children, as we determined their number in the first search. When it happens, when $u$ has received the information from all its non-hollow children, $u$ is able to determine its own information: $u$ makes the sum of $|V(v)\cap N(x)|$ for all its non-hollow children $v$, and $u$ determines whether it is completion-forced as follows. If $u$ is parallel, then $u$ is completion-forced iff all its children are non-hollow, and if $u$ is series, then $u$ is completion-forced iff all its children are completion-forced. Then, $u$ forwards its information to its parent and the process goes on until the root of the tree itself has determined its information. At that time, the process ends as all the non-hollow nodes of $T$ have already determined their information.

\subsubsection{Second step: finding all completion-minimal insertion nodes of $\mathbf{T}$.}

We search the set of all non-hollow, eligible and non-completion-forced nodes of $T$. For each of them, we determine whether it is a minimal insertion node and, in the positive, we compute the number of edges to be added in its associated minimal completion. Then, at the end of the search we select the completion of minimum cardinality.

Since, all the ancestors of a non-hollow eligible non-completion-forced node also satisfy these three properties, it follows that the part of $T$ we have to search is a connected subset of nodes containing the root. Then, our search starts by determining whether the root is non-completion-forced. In the negative, we are done: there exists one unique minimal completion of $G+x$ which is obtained by adding all missing edges between $x$ and the vertices of $G$.

Otherwise, if the root is non-completion-forced (it is always eligible, by definition, and non-hollow, from Remark~\ref{rem:rootnofull}), we start our search. For all the non-hollow children of the current node (we built their list in the first step), we check whether they are eligible and non-completion-forced and search, in a depth-first manner, the subtrees of those for which the test is positive (cf. Lemma~\ref{lem:explore}).

During this depth-first search, we compute for each node $u$ encountered the number of edges, denoted $cost-above(u)$, to be added between $x$ and the vertices of $V(G)\setminus V(u)$ in the completion anchored at $u$. This can be computed during the search as follows:
\begin{itemize}
\item if the parent $v$ of $u$ is a parallel node (necessarily eligible, since we parse only eligible nodes), then $cost-above(u)=cost-above(v)$; and
\item if the parent $v$ of $u$ is a series node, then $cost-above(u)=cost-above(v)+\sum_{u'\in\cur{C}(v), u'\neq u} |V(u')\setminus N(x)|$.
\end{itemize}

We also determine whether $u$ is a minimal insertion node by testing whether it satisfies Condition~\ref{ser} or~\ref{par} of Lemma~\ref{lem:stop}. This can be done thanks to the information collected in the first step. Importantly for the complexity, note that Condition~\ref{par} of Lemma~\ref{lem:stop} can be tested by scanning only the non-hollow children of $u$.
In the positive, if $u$ is a minimal insertion node, then we determine the number of edges, denoted $cost(u)$, to be added in the completion anchored at $u$, as $cost(u)=cost-above(u)+\sum_{v\in\cur{C}_{nh}(u)}|V(v)\setminus N(x)|$.

From Lemma~\ref{lem:stop}, minimal insertion nodes are non-hollow, eligible and non-completion-forced. Therefore, our search discovers all the completion-minimal insertion nodes, and computes the cost of their associated minimal completion. We keep track of the minimum cost completion encountered during the search and outputs the corresponding insertion node at the end.
Finally, we need to update the cotree $T$ for the next incremental step of the algorithm (as depicted in Figure~\ref{fig:modif-tree}). To this purpose, we use the algorithm of~\cite{CPS85} as explained below. 

\subsubsection{Complexity.}

The key of the $O(d')$ time complexity is that we search and manipulate only the set $\cur{NH}(T)$ of non-hollow nodes of $T$. For each of them $u$, we need to scan the list of its non-hollow children $\cur{C}_{nh}(u)$ and to perform a constant number of tests and operations that can all be done in $O(1)$ time (thanks to the information collected in the first step). For example, when we need to test the number of hollow children of $u$ we avoid to count them by computing their number as $|\cur{C}(u)|-|\cur{C}_{nh}(u)|$. The computation of $cost-above(u)$ can also be done in $O(1)$ time by noting that the sum $\sum_{u'\in\cur{C}(v), u'\neq u} |V(u')\setminus N(x)|$ can rather be computed as $|V(v)\setminus N(x)|-|V(u)\setminus N(x)|$. Therefore, treating a node $u$ takes time $O(|\cur{C}_{nh}(u)|$ and the execution of the two steps of the algorithm takes $O(|\cur{NH}(T)|)$ time.

Furthermore, as shown in~\cite{LMP10}, we have $|\cur{NH}(T)|=O(d')$, where $d'=|N'(x)|$ is the cardinality of the completed neighbourhood of $x$. Indeed, from Theorem~\ref{th:insert}, all non-hollow nodes are filled except the ancestors of the insertion node $u$. Let $v$ be a non-hollow child of one ancestor of $u$, then $v$ is filled and it follows that the sum of the sizes of $T_v$ for all such $v$ is bounded by $N'(x)$. The number of ancestors of $v$ is also linearly bounded by $N'(x)$ as half of these ancestors are series and therefore have a child $v$ which is filled.

When, the insertion node $u$ has been determined, the completed neighbourhood $N'(x)$ of $x$ can be computed in extension by a search of the part of $T$ that is filled, which takes $O(d')$ time. Then, the cotree of the completion $H$ of $G+x$ is obtained from the cotree of $G$ (as depicted in Figure~\ref{fig:modif-tree}) in the same time complexity thanks to the algorithm of~\cite{CPS85}. Overall, one incremental step takes $O(d')$ time and the whole running time of the algorithm is $O(n+m')$, where $m'$ is the number of edges in the output cograph.

\section{An $\mathbf{O(n+m\, log^2 n)}$ algorithm}\label{sec:improv-algo}

Even though it is linear in the number of edges in the output cograph, the $O(n+m')$ complexity achieved by the algorithm in~\cite{LMP10} and the one we presented in Section~\ref{sec:algo-compl} is not necessarily optimal, as the output cograph can actually be represented in $O(n)$ space using its cotree. We then design a refined version of the inclusion-minimal completion algorithm that runs in $O(n+m\log^2 n)$ time, when no additional condition is required on the completion output at each incremental step. This improvement is further motivated by the fact that, as we show below, there exist graphs having only $O(n)$ edges and which require $\Omega(n^2)$ edges in any of their cograph completions. For such graphs, the new complexity we achieve also writes $O(n\log^2 n)$ (since $m=O(n)$) and constitutes a significant improvement over the $O(n^2)$ complexity of the previous algorithm (since $m'=\Omega(n^2)$).

\subsection{Worst-case minimum-cardinality completion of very sparse graphs}

In this section, we show that there exist graphs that have only $O(n)$ edges and that require $\Omega(n^2)$ edges in any of their cograph completions. Actually, we show that this even holds in the more general case where the target graph class has bounded rank-width (see~\cite{OS06} for a definition), which includes the class of cographs as well as the class of distance hereditary graphs (see~\cite{S03} for a definition). Furthermore, although it is not necessary for the purpose of this article, we also show that the same behaviour occurs for chordal completion, as we believe that this fact is interesting in itself.
Our proofs are based on the notion of vertex expander graphs (see~\cite{hoory2006expander} for a survey on the topic). We first show that these graphs require $\Omega(n^2)$ edges in any of their cograph completions, as stated by Theorem~\ref{th:compl-brw} below, and we conclude by pointing out that there exist constructions of vertex expander graphs with only $O(n)$ edges.

\begin{definition}[Vertex expander]
A graph $G$ is a \emph{$c$-expander} if, for every vertex subset $S \subseteq V(G)$ with $|S| \leq \frac{|V(G)|}{2}$ we have $|N(S)| \geq c \cdot |S|$.
\end{definition}

In our proof of Theorem~\ref{th:compl-brw}, we will use the fact that cographs are graphs of bounded rank-width, for which we have Proposition~\ref{prop:balsep} below. Roughly speaking, it states that if a graph $G$ has rank-width at most $r$, then there exists a cut of $G$ of rank at most $r$ such that both parts of the cut are large. 

\begin{proposition}[\cite{OumS07}]\label{prop:balsep} Let $r$ be an integer and let $G$ be a graph whose rank-width is at most $r$. Then there exists a subset $S \subseteq V(G)$ of vertices, such that $\frac{n}{3} \leq |S| \leq \frac{n}{2}$ and $\cutrank(S) \leq r$.
\end{proposition}

We remark that Proposition~\ref{prop:balsep} is stated by Oum and Seymour~\cite{OumS07} in terms of symmetric submodular functions. Also see~\cite{OS06} for definitions of \emph{rank-width} and \emph{cutrank}. We will need the following proposition which shows that if a cut $(S,V\setminus S)$ of a graph has a small rank, say $r$, then there can be only a small number\footnote{More explicitly, this number is bounded by a quantity depending only on $r$.} of equivalence classes of vertices in $S$ according to their neighbourhood in $V\setminus S$.

\begin{proposition}[\cite{vatshelle2012new}]\label{prop:neighequiv} Let $G$ be a graph and $S \subseteq V(G)$ be a vertex set such that $\cutrank(S) \leq r$. Then there exists a partition $S = S_1 \cup S_2 \cup \ldots \cup S_t$, $t \leq 2^r$ such that for every $i \leq t$ and pair $u,v$ of vertices in $S_i$, $N(u) \setminus S = N(v) \setminus S$.
\end{proposition}

We are now ready to state and prove Theorem~\ref{th:compl-brw}, regarding completions in graph classes $\cur{H}$ of bounded rank-width.

\begin{theorem}\label{th:compl-brw}
Let $c > 0$ be a positive real number and $r$ be a positive integer. Let also $G$ be a $c$-expander and $\cur{H}$ be a class of graphs whose rank-width is at most $r$. Then, there exists a positive real number $K_{c,r}$, depending only on $c$ and $r$, such that any completion of $G$ into a graph in $\cur{H}$ has at least $K_{c,r} \cdot n^2$ edges.
\end{theorem}

\begin{proof}
Let $H$ be a completion of $G$ into a graph in $\cur{H}$. Since $H$ is a supergraph of $G$, it follows immediately from the definition that $H$ is a $c$-expander. Moreover, since $H$ has rank-width at most $r$, from Propositions~\ref{prop:balsep} and~\ref{prop:neighequiv}, there exists a subset $S \subseteq V(G)$ of vertices, such that $\frac{n}{3} \leq |S| \leq \frac{n}{2}$ and there exists a partition $S = S_1 \cup S_2 \cup \ldots \cup S_t$, with $t \leq 2^r$, such that for every $i \leq t$ and any pair $u,v$ of vertices in $S_i$, $N(u) \setminus S = N(v) \setminus S$. Assume, without loss of generality, that the $S_i$'s are ordered by increasing cardinality. We denote $U_i=\bigcup_{j\in\int{1,i}} S_i$.

If $|S_1|> \frac{c}{2}\, |S\setminus S_1|$, then we have $|S_1|> \frac{c}{2}\, |S|- \frac{c}{2}\, |S_1|$ and so $(1+\frac{c}{2})\, |S_1|> \frac{c}{2}\, |S|$, which gives $|S_1|> \frac{c}{2+c}\, |S|$. And since the $S_i$'s are ordered by increasing size, we conclude that the inequality holds for all indices: for all $i\in\int{1,t}$, we have $|S_i|> \frac{c}{2+c}\, |S|$.

In the complement case, i.e. if $|S_1|\leq \frac{c}{2}\, |S\setminus S_1|$, then consider the largest index $i$ such that $|U_i|\leq \frac{c}{2}\, |S\setminus U_i|$. Note that necessarily we have $1\leq i<t$. We now prove that $|S_{i+1}|=\Omega(|S|)$, where the hidden factor depends only on $c$ and $r$. By definition of $i$, we have $|U_i|+|S_{i+1}|> \frac{c}{2}\, (|S|-|U_i|-|S_{i+1}|)=\frac{c}{2}\, |S|-\frac{c}{2}\, (|U_i|+|S_{i+1}|)$. This gives $(1+\frac{c}{2}\,) (|U_i|+|S_{i+1}|)> \frac{c}{2}\, |S|$. On the other hand, because the $S_i$'s are ordered by increasing cardinality, we have that $|U_i|\leq i\, |S_{i+1}|\leq 2^r\, |S_{i+1}|$. By injecting this inequality in the one above we obtain $(1+\frac{c}{2})\, (1+2^r)\, |S_{i+1}|>\frac{c}{2}\, |S|$ and so $|S_{i+1}|> \frac{c}{(2+c)\, (1+2^r)}\, |S|$.

As a partial conclusion, we have either \emph{(i)} for all $i\in\int{1,t}$, $|S_i|> \frac{c}{2+c}\, |S|$, or \emph{(ii)} there exists $i\in\int{1,t-1}$ such that $|U_i|\leq \frac{c}{2}\, |S\setminus U_i|$ and for all $j\in\int{i+1,t}$, we have $|S_j|> \frac{c}{(2+c)\, (1+2^r)}\, |S|$ (because the $S_i$'s are ordered by increasing cardinality). Beside this, because of the expansion property of $S$, we have $|N(S)| \geq c\, |S|$, meaning that there are at least $c \cdot |S|$ vertices out of $S$ that are adjacent to at least one vertex of $S$. Moreover, note that from the definition of the $S_i$'s, we have that if a vertex $x\in V\setminus S$ is adjacent to some vertex $y\in S_i$, for some $i\in\int{1,t}$, then $x$ is adjacent to all the vertices of $S_i$. In case \emph{(i)} of the alternative above, where $|S_i|> \frac{c}{2+c}\, |S|$ for all $i\in\int{1,t}$, we obtain that there must be at least $c\, |S| \cdot \frac{c}{2+c}\, |S|=\frac{c^2}{2+c}\, |S|^2$ edges between $S$ and $V\setminus S$ in graph $H$. Thus, in this case, because $|S|\geq \frac{n}{3}$, the conclusion of the theorem holds.

In the other case, i.e. case \emph{(ii)} of the alternative above, we have $|U_i|\leq \frac{c}{2}\, |S\setminus U_i|$ for some $i\in\int{1,t-1}$ and for all $j\in\int{i+1,t}$, $|S_j|> \frac{c}{(2+c)\, (1+2^r)}\, |S|$. The expansion property applied to $S\setminus U_i$ gives $|N(S\setminus U_i)|\geq c \, |S\setminus U_i|$. Since $|U_i|\leq \frac{c}{2}\, |S\setminus U_i|$, we have $|N(S\setminus U_i)\setminus S| \geq \frac{c}{2}\, |S\setminus U_i|$. Observe that because $S_{i+1}\subseteq |S\setminus U_i|$, we have $|S\setminus U_i|\geq |S_{i+1}|\geq \frac{c}{(2+c)\, (1+2^r)}\, |S|$ and consequently $|N(S\setminus U_i)\setminus S| \geq \frac{c^2}{2\, (2+c)\, (1+2^r)}\, |S|$. Moreover, each of the vertices in $N(S\setminus U_i)\setminus S$ is adjacent to all the vertices of $S_j$ for some $j\in\int{i+1,t}$. And since $|S_j|> \frac{c}{(2+c)\, (1+2^r)}\, |S|$, we obtain that there are at least $\frac{c^2}{2\, (2+c)\, (1+2^r)}\, |S| \cdot \frac{c}{(2+c)\, (1+2^r)}\, |S|=\frac{c^3}{2\, (2+c)^2\, (1+2^r)^2}\, |S|^2\geq \frac{c^3}{18\, (2+c)^2\, (1+2^r)^2}\, n^2$ edges between $S\setminus U_i$ and $N(S\setminus U_i)\setminus S$ in graph $H$ (because $|S|\geq \frac{n}{3}$), which achieves the proof of the theorem.
\end{proof}

\begin{remark}
The result of Theorem~\ref{th:compl-brw} holds in particular for cographs and distance hereditary graphs, which both have rank-width at most $1$.
\end{remark}

It is also worth noting that in the particular cases of cographs and distance hereditary graphs, the proof above can be greatly simplified as follows. For a cut $(S,V\setminus S)$ of rank at most $1$, all the vertices of $S$ having some neighbour in $V\setminus S$ have exactly the same neighbours in $V\setminus S$. This corresponds to the fact that there are at most $2$ equivalence classes $S_1,S_2$ in Proposition~\ref{prop:neighequiv} ($r=1$): the vertices of $S$ that have some neighbour in $V\setminus S$ and those that do not have any. Moreover, the expansion property for $S$ and for $V\setminus S$ (remind that from Proposition~\ref{prop:balsep} we have $\frac{n}{3} \leq |S| \leq \frac{n}{2}$) implies that the numbers of vertices in $S$ and in $V\setminus S$ that have some neighbour on the other side of the cut are both $\Omega(c.n)$, which proves the statement of Theorem~\ref{th:compl-brw}.

The results above hold for any input graph that is a $c$-expander. Nevertheless, in order to achieve our goal, we still need the existence of very sparse $c$-expanders. This has already been established as there exist deterministic constructions of very sparse graphs that are $c$-expanders, see for example the construction of $3$-regular $c$-expanders by Alon and Boppana~\cite{Alo86}, for some fixed $c$. Such graphs have only $O(n)$ edges but, from Theorem~\ref{th:compl-brw}, require $\Omega(n^2)$ edges in any of their cograph completions (as well as in any of their completions in a graph class $\cur{H}$ of bounded rank-width). More generally, it is part of the folklore that, for any constant $a > 1$, there exist $c>0$ and $p>0$ such that, for any $n\in \mathbb{N}$ sufficiently large, the proportion of graphs on $n$ vertices and $a.n$ edges that are $c$-expanders is at least $p$. This means that many graphs of fixed mean degree have the vertex expansion property and therefore require $\Omega(n^2)$ edges in any of their cograph completions. Motivated by this frequent worst-case for the $O(n+m')$ complexity, we will design an $O(n+m\log^2 n)$-time algorithm for inclusion-minimal cograph completion of arbitrary graphs.

\subsubsection{A similar behaviour for chordal completion.}

The fact that some very sparse graphs, having $O(n)$ edges, may require $\Omega(n^2)$ edges in any of their completions also occurs for other target classes, whose rank-width is unbounded. In particular, we now show that the very popular chordal completion problem also exhibits such a behaviour, which we believe is worth of interest in itself, though unnecessary for the strict purpose of this article. Our proof is as previously based on vertex expander graphs, for which we have the following result.

\begin{proposition}[\cite{FlumG06}]\label{prop:bigtw} If $G$ is a $c$-expander for a constant $c > 0$ independent of $n$, then the treewidth of $G$ is $\Omega(n)$.
 \end{proposition}

In addition, it is well known (see~\cite{ACP87}) that the treewidth of a graph $G$ is the minimum size (minus $1$) of the maximum clique among all chordal completions of $G$. Consequently, Proposition~\ref{prop:bigtw} immediately gives an $\Omega(n^2)$ lower bound on the number of edges in any chordal completion $H$ of a $c$-expander $G$, since $H$ must have a clique of size $\Omega(n)$. To conclude, remind that, as mentioned above, there exist constructions, both deterministic and random, of $c$-expanders having only $O(n)$ edges.

\smallskip
We now turn to the description of our $O(n+m\log^2 n)$-time algorithm for inclusion-minimal cograph completion.

\subsection{Data structure}

Our data structure is composed of two copies of the cotree: one stored in a basic data structure and one using the advanced dynamic data structure of~\cite{ST83} named \emph{dynamic trees}.
We note that we could use only the advanced data structure of~\cite{ST83}, as it can be patched to contain the additional information that we store in the basic data structure. But to avoid questions about the compatibility of such a patch with the performances of the data structure of~\cite{ST83}, we prefer to store the additional information we need, and to perform the related operations, independently in another structure. This is the reason why we describe our algorithm using two structures.

In the first copy of $T$ (the basic data structure), each node $u$ stores its parent, the list of the children of $u$ and their number $|\cur{C}(u)|$, as well as a bidirectional couple of pointers to the corresponding node of $u$ in the second copy of $T$, so that we can move from one element in one copy of the cotree to the same element in the other copy in $O(1)$ time. In addition, we enhance this basic data structure storing the cotree with one additional feature: given a node $u$ and two of its children $u_1,u_2$, this feature allows us to determine which of $u_1,u_2$ appears first in the list of children of $u$ in $O(1)$ time. To this purpose, the set of children of a node $u$ is not only stored in a doubly linked list, as in the classical version of the tree, but a copy of this list is also stored using the \emph{order data structure} of~\cite{BCD+02,DS87}. This data structure allows to answer \emph{order queries}, i.e. which of two given elements of the list precedes the other one, and supports two update operations, \texttt{insert} and \texttt{delete}. The delete operation removes a given element from the order data structure while the insert operation insert a new element in the order data structure just after a specified element. The order query and the two update operations all take $O(1)$ worst-case time.

\paragraph*{Dynamic trees~\cite{ST83}}
In addition to the classical data structure described above, we also use the data structure developed in~\cite{ST83} to store a copy of the cotree $T$ and maintain it at each incremental step. This data-structure maintains a dynamic forest rather than a single tree. This will be useful for us as we will cut a part of the cotree and attach it to another node during the update of the cotree under the insertion of a new vertex. The dynamic trees of~\cite{ST83} allow to answer the two following kinds of query:

\begin{description}
\item[\texttt{lowest-common-ancestor?}] Given two nodes $u$ and $v$ of $T$, provide the lowest common ancestor $lca(u,v)$ of $u$ and $v$.

\item[\texttt{next-step-to-descendant?}] Given a node $u$ of $T$ and one of its strict descendants $v$, provide the (unique) child of $u$ which is an ancestor of $v$.
\end{description}

These two kinds of query are handled in $O(\log n)$ worst-case time in the data structure of Sleator and Tarjan~\cite{ST83}. To be precise, the second operation is not described in~\cite{ST83}, but it can be obtained as a combination of other operations they provide. Indeed, their data structure also supports, in the same complexity:
\begin{itemize}
\item an update operation called \texttt{evert}$(u)$ which, given a vertex $u$ of $T$, makes $u$ become the root of $T$, and
\item a query operation named \texttt{root?}$(u)$ that provides the root of the tree $T$ to which node $u$ belongs.
\end{itemize}

Then, the query \texttt{next-step-to-descendant?}$(u,v)$ we use here can be resolved by the sequence of operations (two updates and two queries): $r=$\texttt{root?}$(u)$, \texttt{evert}$(v)$, \texttt{parent?}$(u)$, \texttt{evert}$(r)$, which takes $O(\log n)$ time. 

Along our incremental algorithm, we need to maintain the dynamic data structure of~\cite{ST83}, which can be done thanks to the following update operations:

\begin{description}
\item[\texttt{cut}.] Given a node $u$ in a tree $T$ of the forest $F$ such that $u$ is not the root of $T$, remove the edge between $u$ and $parent(u)$. Then, $u$ becomes the root of its new tree $T'$ in forest $F$.

\item[\texttt{link}.] Given a node $u$ in a tree $T$ of the forest $F$ such that $u$ is not the root of $T$ and given the root $v$ of a tree $T'\neq T$, make $u$ the parent of $v$.
\end{description}

Note that operations \texttt{cut} and \texttt{link} are converse of each other. As for queries, all update operations takes $O(\log n)$ worst-case time.

\subsection{Algorithm}

Our algorithm determines the set $W$ of the nodes that are simultaneously eligible, non-hollow and non-completion-forced and that are minimal for the ancestor relationship among nodes having these three properties (i.e. none of their descendants satisfies the considered property). Then, it picks any of them to be the insertion node of the minimal completion returned at this incremental step. Indeed, since nodes of $W$ satisfy the conditions of Lemma~\ref{lem:explore} and none of their children does (because nodes of $W$ are minimal for the ancestor relationship), it follows that nodes of $W$ are completion-minimal insertion nodes. In order to get the improved $O(n+m\log^2 n)$ complexity, we avoid to completely search the upper tree $T^{up}_{N(x)}$ to determine $W$. Instead, we use a limited number of \texttt{lowest-common-Ancestor?} queries.

Clearly, if a parallel node $u$ of $T$ is the $lca$ of two leaves in $N(x)$ then $T_u\setminus\set{u}$ contains no eligible node. Let $P_{max}$ be the set of parallel common ancestors of vertices of $N(x)$ that are maximal for the ancestor relationship and let us denote $W'=P_{max}\cup N_{out}$, where $N_{out}$ is the set of vertices of $N(x)$ that are not descendant of any node in $P_{max}$, i.e. $N_{out}=N(x)\setminus \bigcup_{p\in P_{max}} V(p)$. Note that all the nodes $w'\in W'$ are eligible, and so are their ancestors. It follows that the set $W$ that we want to compute is the set of the non-completion-forced nodes in the upper tree $T^{up}_{W'}$ that are minimal for the ancestor relationship (i.e. none of their descendants in $T^{up}_{W'}$ are not completion-forced).

\subsubsection{Finding an inclusion-minimal insertion node.}

In order to compute $W$, we start by computing the tree $\widetilde{T}=T^{xtr}_{N(x)\cup A_x}$ extracted from (see Section~\ref{sec:prel}) the leaves that belong to $N(x)$ and the set $A_x$ of their lowest common ancestors, i.e. nodes $u$ such that $u=lca(l_1,l_2)$ for some leaves $l_1,l_2\in N(x)$. Then, we search $\widetilde{T}$ to find its parallel nodes $P_{max}$ that are maximal for the ancestor relationship and we remove their strict descendants. The leaves of the resulting tree are exactly nodes of $W'$. Finally, for each node $w'\in W'$ we determine its lowest non-completion-forced ancestor $nfa(w')$ in $T$ and we keep only the $nfa(w')$'s that are minimal for the ancestor relationship: this is the set $W$. It is worth noting from the beginning that since $\widetilde{T}$ has exactly $d$ leaves and since all its internal nodes have degree at least $2$, then the size of $\widetilde{T}$ is $O(d)$.

Let us now show how to compute $\widetilde{T}$ in $O(d\log^2 n)$ time. To this purpose, we sort the neighbours of $x$ according to a special order of the vertices of the cograph $G$ called a \emph{factorising permutation}~\cite{CHM02}. A factorising permutation is the order in which the vertices of $G$ (which are the leaves of the cotree) are encountered when performing a depth-first search of the cotree $T$. There are as many different factorising permutations as different depth-first search of $T$. Here, we use the factorising permutation $\pi$ which is obtained by visiting the children of one node $u$ of $T$ in the order they are stored in the list of the children of $u$ used in the implementation of the cotree. To determine whether a vertex $y_1$ is before or after a vertex $y_2$ in the factorising permutation $\pi$, we can proceed as follows: 1) find $u=lca(y_1,y_2)$ and find the two children $u_1$ and $u_2$ of $u$ that are respectively the ancestor of $y_1$ and $y_2$, and 2) determine whether $u_1$ is before or after $u_2$ in the list of children of $u$. Operation 1) can be executed in $O(\log n)$ time thanks to the data structure of~\cite{ST83} by performing one \texttt{lowest-common-ancestor?} query and two \texttt{next-step-to-descendant?} queries. Operation 2) can be executed in $O(1)$ time using the order data structure of~\cite{BCD+02,DS87}. Then, comparing the order of occurrence of two vertices $y_1$ and $y_2$ in $\pi$ takes $O(\log n)$ time and totally, sorting all the neighbours of $x$ respectively to order $\pi$ takes $O(d\log d\log n)=O(d\log^2 n)$ time.

The benefit of doing so is that, once the neighbours of $x$ are sorted in the order $x_1,x_2,\ldots, x_d$ in which they appear in $\pi$ (we say from left to right), we can build $\widetilde{T}$ efficiently. We consider the neighbours of $x$ one by one in this order and at each step we compute the tree $T_i$ extracted from $\set{x_1,\ldots,x_{i}}$ and their lowest common ancestors. Then, at the end of the computation, when $i=d$, we obtain $T_d=\widetilde{T}$. For each $i$ between $2$ and $k$, we obtain $T_i$ from $T_{i-1}$ as follows: we compute $v_i=lca(x_{i-1},x_i)$ and we insert it at its correct position in the tree $T_{i-1}$ built so far.

Note that, since we consider the $x_i$'s from left to right in the order of the factorising permutation $\pi$, the newly computed common ancestor $v_i$ is the only node that may be in $T_i$ but not in $T_{i-1}$. Moreover, for the same reason, if $v_i$ is not yet a node of $T_{i-1}$ then $v_i$ has to be inserted on the rightmost branch of the tree $T_{i-1}$, and if $v_i$ is already a node of $T_{i-1}$ then $v_i$ already belongs to this branch, and so we discover it when we try to insert it on this branch. In order to do so, we climb up the rightmost branch of $T_i$, starting from the father of $x_{i-1}$, and for each node $v$ encountered on this branch we determine whether $v_i$ is higher or lower than $v$ in the tree (or eventually equal) by computing $lca(v,v_i)$. The total number of comparisons (treated by $lca$ queries) made along the computation of $T_d$ is $O(d)$. Indeed, as explained in~\cite{GBT84}, every time we pass above a node $v$ on the rightmost branch, $v$ leaves the rightmost branch for ever and will then never participate again to any comparison. Then, the total number of $lca$ queries we need to built $\widetilde{T}$ (including the $d-1$ queries made on the pairs of neighbours of $x$ appearing consecutively in the order of the factorising permutation) is proportional to its size, that is $O(d)$. Since each of these queries takes $O(\log n)$ time thanks to the data structure of~\cite{ST83}, the complexity of building $\widetilde{T}$ from the sorted list of neighbours of $x$ is $O(d\log n)$.

Once $\widetilde{T}$ is built, a simple search starting from its root determines the set $P_{max}$ of its parallel nodes that are maximal for the ancestor relationship, and we cut off from $\widetilde{T}$ all the subtrees rooted at the children of nodes in $P_{max}$. The leaves of the resulting tree are precisely the nodes of $W'$ that we wanted to determine. As $\widetilde{T}$ has size $O(d)$, this step takes $O(d)$ time.

Then, for each $w'\in W'$, we determine its lowest non-completion-forced ancestor $nfa(w')$ in $T$. From the definition of $P_{max}$, the lowest parallel ancestor of $w'$ is non-completion-forced. Then, $nfa(w')$ cannot be higher in $T$ than the grand-parent of $w'$. It follows that we have to check the non-completion-forced condition only for $w'$ and its parent, which can be done, for each of them $u$, in $O(|\cur{C}_{nh}(u)|)$ time. Then, we remove the $nfa(w')$'s that are not minimal for the ancestor relationship to obtain the set $W$, this takes $O(d)$ time, and

We now need to find the non-completion-forced nodes of the upper tree $T^{up}_W$ that are minimal for the ancestor relationships. To that purpose, for each $w'\in W'$, we determine its lowest non-completion-forced ancestor $nfa(w')$ in $T$. From the definition of $P_{max}$, the lowest parallel ancestor of $w'$ is non-completion-forced. Therefore, $nfa(w')$ cannot be higher in $T$ than the grand-parent of $w'$. It follows that we have to check the non-completion-forced condition only for $w'$ and its parent, which can be done as follows. If $w'$ is a leaf of $T$, i.e. $w'\in N_{out}$, then $w'$ is forced. If $w'$ is a parallel node of $T$, i.e. $w'\in P_{max}$, then $w'$ is forced iff its number of children in $\widetilde{T}$ equals its number of children in $T$. Now, for the parent $v$ of $W'$, if $v$ is a parallel node, as noted above, $v$ is necessarily non-completion-forced. Otherwise, if $v$ is a series node, $v$ is completion-forced iff i) $v$ belongs to $\widetilde{T}$ and ii) its number of children in $\widetilde{T}$ equals its number of children in $T$, and iii) all its children in $T$ belong to $W'$ and iv) all its children in $T$ are forced (cf. conditions given above for $w'$). If none of $w'$ and $parent(w')$ is non-completion-forced, then necessarily $parent(parent(w'))$ is. As testing these conditions for one node $u$ takes $O(|\cur{C}_{nh}(u)|)$ time, then determining the $nfa(w')$'s for all nodes $w'\in W'$ takes $O(d)$ time. Finally, we determine the $nfa(w')$'s that are minimal for the ancestor relationship, i.e. the nodes of $W$, by searching $T$ upward on at most two levels, starting from each of the nodes in $W'$. This also takes $O(d)$ time. Then, we arbitrarily pick one node $w$ in $W$ and the minimal completion of the neighbourhood of $x$ returned by the algorithm is the one anchored at $w$. Therefore, the total complexity of finding one completion-minimal insertion node in one incremental step of the algorithm is $O(d+d\log n+d\log^2 n)=O(d\log^2 n)$.

\begin{figure}
\begin{center}
\resizebox{.95\linewidth}{!}{
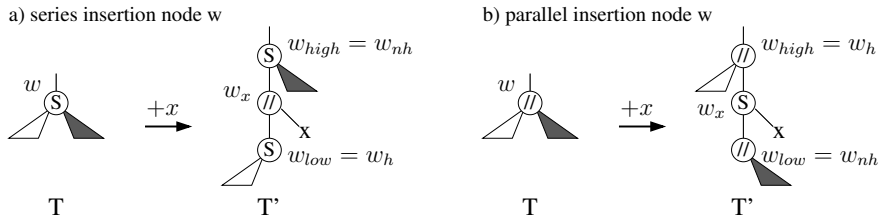
}
\caption{Modification of the cotree under the insertion of $x$ at insertion node $w$. The triangles in black (resp. white) correspond to the parts of the tree that are filled (resp. that remain hollow) in the completion anchored at $w$.}
\label{fig:modif-tree}
\end{center}
\end{figure}

\subsubsection{Updating the data structure.}

In the previous section, we showed how to determine the insertion node $w$ and the list of its children to be filled. Then, depending on whether $w$ is a parallel or a series node, the cotree $T$ must be modified as shown in Figure~\ref{fig:modif-tree}, and the data structure of~\cite{ST83} associated to the cotree must be updated accordingly. The key for doing so while preserving the $O(d\log^2 n)$ time complexity is to perform operations involving only the non-hollow children of $w$. Indeed, their number is $O(d)$, while the number of the hollow children of $w$ can be up to $\Omega(n)$ and arbitrarily large compared to $d$.

After the insertion of $x$, the insertion node $w$ is replaced by three nodes, see Figure~\ref{fig:modif-tree}. Two of them have the same label as $w$: one $w_h$ has for children the hollow children of $w$ and the other one $w_{nh}$ has for children the non-hollow children of $w$. In order to preserve the complexity, it is important to form these two nodes as follows. We cut from $w$ its non-hollow children and its parent, we then obtain $w_h$, still linked to its correct children. Then, we link all the non-hollow children of $w$ to a new node $w_{nh}$. This takes $O(d\log n)$ as it requires $O(d)$ \texttt{cut} and \texttt{link} operations, each of which is supported in $O(\log n)$ time by the data structure of~\cite{ST83}, and the corresponding \texttt{delete} and \texttt{insert} operations in the order data structure storing the lists of children of the nodes in the tree take $O(1)$ time. The rest of the transformations in order to get the new tree as depicted in Figure~\ref{fig:modif-tree} only requires 4 \texttt{link} operations. Thus, the time complexity of updating the data structure in one incremental step is $O(d\log n)$.

As a conclusion, the complexity of one incremental step of the algorithm is $O(d\log^2 n)$ and overall, the complexity of the whole algorithm is $O(n+m\log^2 n)$.

\section{Conclusion and perspectives}\label{sec:conclu}

We designed two incremental algorithms for computing an inclusion-minimal cograph completion of an arbitrary graph $G$. The first one has a time complexity of $O(n+m')$, where $m'$ is the number of edges in the output completion, which matches the complexity of \cite{LMP10}'s algorithm. The specificity of our algorithm is that, within this complexity, it is able to compute a minimum-cardinality completion of the neighbourhood of the new vertex $x$ introduced at each step of the incremental algorithm, which is a highly desirable feature in practice to obtain inclusion-minimal completions of small cardinality. The way we achieved this is by scanning at each incremental step the set of all possible minimal completions of the neighbourhood of $x$. This is particularly interesting as, beside the minimum-cardinality criteria, this opens the possibility of choosing the completion selected by the algorithm using any criteria one wishes.

Our second algorithm improves the time complexity of computing an inclusion-minimal cograph completion of $G$ to $O(n+m\log^2n)$. This improvement is motivated by the fact that, as we gave evidence for it, many graphs (namely those having the expansion property) that have only $O(n)$ edges require $\Omega(n^2)$ edges in any of their cograph completions. Unfortunately, we obtained this improved complexity at the price of giving up the additional feature obtained in the first algorithm, namely computing a minimum-cardinality completion of the neighbourhood of $x$ at each incremental step. Therefore, the first open question arising from our work is whether it is possible to provide this functionality within the $O(n+m\log^2n)$ time complexity, or at least within a time complexity of the form $O(n+m\, \text{polylog}(n))$.

The question of improving further the time complexity, when expressed with regard to the size of the input graph, of computing an inclusion-minimal cograph completion is also open. Although it seems difficult to reach a linear complexity with the techniques we use here, nothing indicates that the $O(n+m\log^2n)$ complexity we obtained could not be improved further, say for example to $O(n+m\log n)$. Such an improvement would be very valuable both in theory and in practice for dealing with very large real-world networks~\cite{JGG+15,HWL+15,Cre17}.

Another appealing perspective is to design algorithms that are able to use not only addition of edges but also deletion of edges in order to minimally modify an arbitrary graph into a cograph. What is the best complexity that can be achieved for the general cograph editing problem (where both addition and deletion of edges are allowed)? Is it possible, in this case as well, to design an incremental algorithm that provides a minimum-cardinality modification of the neighbourhood of $x$ at each incremental step? The behaviour of the general cograph editing problem seems quite different from the one of the pure completion (or pure deletion) problem. Therefore, answering these questions would significantly contribute to leverage our understanding of graph modification problems.



\bibliographystyle{splncs03}
\bibliography{mod_cographLinearity,mod-cograph-completion,MIC}

\end{document}

%% file: modif-tree.pdf_t
\begin{picture}(0,0)%
\includegraphics{modif-tree.pdf}%
\end{picture}%
\setlength{\unitlength}{4144sp}%
\begingroup\makeatletter\ifx\SetFigFont\undefined%
\gdef\SetFigFont#1#2#3#4#5{%
  \reset@font\fontsize{#1}{#2pt}%
  \fontfamily{#3}\fontseries{#4}\fontshape{#5}%
  \selectfont}%
\fi\endgroup%
\begin{picture}(6593,1659)(256,-571)
\put(1441,209){\makebox(0,0)[b]{\smash{{\SetFigFont{12}{14.4}{\rmdefault}{\mddefault}{\updefault}{\color[rgb]{0,0,0}$+x$}%
}}}}
\put(5041,209){\makebox(0,0)[b]{\smash{{\SetFigFont{12}{14.4}{\rmdefault}{\mddefault}{\updefault}{\color[rgb]{0,0,0}$+x$}%
}}}}
\put(451,389){\makebox(0,0)[b]{\smash{{\SetFigFont{12}{14.4}{\rmdefault}{\mddefault}{\updefault}{\color[rgb]{0,0,0}$w$}%
}}}}
\put(2116,344){\makebox(0,0)[rb]{\smash{{\SetFigFont{12}{14.4}{\rmdefault}{\mddefault}{\updefault}{\color[rgb]{0,0,0}$w_x$}%
}}}}
\put(2386,704){\makebox(0,0)[lb]{\smash{{\SetFigFont{12}{14.4}{\rmdefault}{\mddefault}{\updefault}{\color[rgb]{0,0,0}$w_{high}=w_{nh}$}%
}}}}
\put(2386,-151){\makebox(0,0)[lb]{\smash{{\SetFigFont{12}{14.4}{\rmdefault}{\mddefault}{\updefault}{\color[rgb]{0,0,0}$w_{low}=w_h$}%
}}}}
\put(5986,-151){\makebox(0,0)[lb]{\smash{{\SetFigFont{12}{14.4}{\rmdefault}{\mddefault}{\updefault}{\color[rgb]{0,0,0}$w_{low}=w_{nh}$}%
}}}}
\put(4051,389){\makebox(0,0)[b]{\smash{{\SetFigFont{12}{14.4}{\rmdefault}{\mddefault}{\updefault}{\color[rgb]{0,0,0}$w$}%
}}}}
\put(5986,704){\makebox(0,0)[lb]{\smash{{\SetFigFont{12}{14.4}{\rmdefault}{\mddefault}{\updefault}{\color[rgb]{0,0,0}$w_{high}=w_h$}%
}}}}
\put(5716,209){\makebox(0,0)[rb]{\smash{{\SetFigFont{12}{14.4}{\rmdefault}{\mddefault}{\updefault}{\color[rgb]{0,0,0}$w_x$}%
}}}}
\end{picture}%